\documentclass[11pt]{article}
\usepackage{fullpage}
\usepackage{amsmath}
\usepackage{amsfonts}
\usepackage{xcolor}
\usepackage{amsthm}
\usepackage[colorlinks=true]{hyperref}
\usepackage{bbold}
\usepackage{lmodern}
\usepackage{slantsc}
\usepackage{doi}

\newtheorem{theorem}{Theorem}
\newtheorem{lemma}[theorem]{Lemma}
\newtheorem{definition}[theorem]{Definition}
\newtheorem{conjecture}[theorem]{Conjecture}
\newtheorem{remark}[theorem]{Remark}

\newcommand{\proofsubparagraph}{\paragraph}

\DeclareMathOperator*{\argmax}{arg\,max}

\newcommand{\Model}{k\mathrm{-OV}_0^\alpha(n)}
\newcommand{\Dist}[1]{k\mathrm{-OV}_{#1}^\alpha(n)}
\newcommand{\Planted}{k\mathrm{-OV}_1^\alpha(n)}
\newcommand{\Plant}{\mathsf{Plant}}
\newcommand{\Pk}{\mathcal{P}_k}

\newcommand{\Olog}{\widetilde{O}}
\newcommand{\polylog}{\operatorname{polylog}}

\newcommand{\OV}[1][U]{\mathbf{#1}}
\newcommand{\inst}{\OV}
\newcommand{\E}{\mathop{\mathbb{E}}}

\newcommand{\yes}{\textsc{yes}}
\newcommand{\no}{\textsc{no}}

\title{The Planted Orthogonal Vectors Problem}
\author{David Kühnemann\thanks{University of Amsterdam. \href{mailto:david.kuhnemann@student.uva.nl}{david.kuhnemann@student.uva.nl}. Part of this work done while visiting Bocconi.} \and Adam Polak\thanks{Bocconi University. \href{mailto:adam.polak@unibocconi.it}{adam.polak@unibocconi.it}.} \and Alon Rosen\thanks{Bocconi University. \href{mailto:alon.rosen@unibocconi.it}{alon.rosen@unibocconi.it}. Work supported by European Research Council (ERC) under the EU’s Horizon 2020 research and innovation programme (Grant agreement No. 101019547) and Cariplo CRYPTONOMEX grant.}}
\date{}

\begin{document}

\maketitle

\begin{abstract}
In the $k$-Orthogonal Vectors ($k$-OV) problem we are given $k$ sets, each containing $n$ binary vectors of dimension $d=n^{o(1)}$, and our goal is to pick one vector from each set so that at each coordinate at least one vector has a zero. It is a central problem in fine-grained complexity, conjectured to require $n^{k-o(1)}$ time in the worst case.

We propose a way to \emph{plant} a solution among vectors with i.i.d.~$p$-biased entries, for appropriately chosen $p$, so that the planted solution is the unique one. Our conjecture is that the resulting $k$-OV instances still require time $n^{k-o(1)}$ to solve, \emph{on average}.

Our planted distribution has the property that any subset of strictly less than $k$ vectors has the \emph{same} marginal distribution as in the model distribution, consisting of i.i.d.~$p$-biased random vectors. We use this property to give average-case search-to-decision reductions for $k$-OV.
\end{abstract}

\section{Introduction}

The security of cryptographic systems crucially relies on heuristic assumptions about average-case hardness of certain computational problems.  Sustained cryptanalysis alongside technological advances
such as large-scale quantum computers, put these hardness assumptions under constant risk of being invalidated. It is therefore desirable to try to design cryptographic schemes based on new computational problems, preferably ones whose hardness is well-studied.

The field of computational complexity developed over the last fifty years a good understanding of the hardness of certain problems -- e.g., SAT is widely believed to require at least superpolynomial, maybe even exponential time~\cite{Williams19} -- however these are worst-case problems, and hence unsuitable for direct use as a basis for cryptography.

Fine-grained complexity~\cite{VVW18} is a younger branch of computational complexity that studies ``hardness of easy problems'', i.e., problems known to be solvable in polynomial time but supposedly not faster than some specified polynomial, say not faster than in cubic time. It gives rise to \emph{fine-grained cryptography}~\cite{BRSV17,LaVigneLW19,Rosen20}, the idea that it might be possible to build cryptography, notably public-key encryption, based on conjectured average-case hardness of polynomial-time problems studied in fine-grained complexity. These problems are easier than NP-hard ones, 
but for polynomials of sufficiently high degree may still be hard enough to give honest parties an adequate advantage over malicious attackers.

\subsection{The k-Orthogonal Vectors Problem}
The Orthogonal Vectors (OV) problem~\cite{Williams05}, together with its generalization $k$-OV, is one of the three main hard problems studied in fine-grained complexity, alongside 3SUM and APSP~\cite{VVW18}. Arguably, among the three, OV is the one whose (worst-case) hardness we understand the most -- in particular because it is implied by the Strong Exponential Time Hypothesis (SETH)~\cite{ImpagliazzoPZ01}, which is about the very well studied SAT problem.

We say that vectors $u_1, ..., u_k \in \{0,1\}^d$ are \emph{orthogonal} if, for all $j \in [d]$, $\prod_{\ell=1}^k u_\ell[j] = 0$, meaning that for every coordinate there is at least one zero entry among the $k$ vectors.
For $k \geq 2$, let ${U_1,...,U_k}$ each be a collection of $n$ $d$-dimensional binary vectors, which we view as matrices in $\{0,1\}^{n\times d}$. We denote by $U_{\ell,i}$ the $i$-th vector of $U_\ell$.
The $k$-Orthogonal Vectors problem ($k$-OV) asks whether there exist $(s_1,\ldots,s_k) \in [n]^k$ such that $U_{1,s_1},U_{2,s_2},\ldots,U_{k,s_k}$ are orthogonal.

\paragraph{Worst-case complexity.}
The naive algorithm solves $k$-OV in time $O(n^kd)$. For any fixed constant $c$, the algorithms by Abboud et al.~\cite{AbboudWY15} and Chan and Williams~\cite{ChanW21} solve OV in dimension $d=c \log n$ in time $O(n^{2-\varepsilon_c})$ with $\varepsilon_c > 0$. However, Gao et al.~\cite{GaoIKW19} conjecture that no such strongly subquadratic algorithm exists for superlogarithmic dimension $d = \omega(\log n)$. This conjecture (known as Low-dimension Orthogonal Vector Conjecture) is also implied by SETH~\cite{Williams05}. Both the upper bound for $d=O(\log n)$ and the SETH-implied hardness for $d=\omega(\log n)$ generalize to $k$-OV, for any constant $k \geq 2$, where the running time barrier is $n^k$~\cite{VVW18}.

\paragraph{Average-case complexity.}
For cryptographic purposes we care about \emph{average-case} hardness --
because we want to be able to efficiently sample instances that are hard to solve (in contrast to only having an existential knowledge that there are some hard instances). Moreover, the sampler shall correctly tell (with good probability) whether its output is a \yes- or \no-instance.

One way to achieve this is to embed a solution in an instance sampled from a distribution that generates \no-instances with good probability. This method of \emph{planting} a solution has been applied to a number of problems, e.g., $k$-Clique~\cite{JuelsP00} and (in the fine-grained setting) $k$-SUM~\cite{DinurKK24,AgrawalSSVV24} (a generalization of 3SUM) and Zero-$k$-Clique~\cite{LaVigneLW19} (a generalization of Zero-Triangle, which is a problem harder than both 3SUM and APSP), but not for ($k$-)OV. The following question remains wide-open~\cite{BRSV17,DalirrooyfardLW20,DalirrooyfardLS25}:
\begin{center}
\emph{How to plant orthogonal vectors (so that they are hard to find)?}
\end{center}

\subsection{Our results}

We propose a way of planting a solution in $k$-OV instances where each vector entry is i.i.d.~according to a $p$-biased\footnote{We say that a random bit is $p$-biased if it equals $1$ with probability $p$ and equals $0$ with probability $1-p$.} coin flip, for an appropriately chosen value of $p$ so that the planted solution is the only one in the instance, with good probability. We conjecture that solving these instances requires $n^{k-o(1)}$ time on average.

Let us remark that all our results are already nontrivial for $k=2$, i.e., for the Orthogonal Vectors problem. However, from the point of view of cryptographic applications, larger values of $k$ are more interesting (as they potentially offer a bigger advantage for the honest parties), so we present all our results in full generality.

\paragraph{Superlogarithmic dimension.}
The $k$-OV problem might have appeared as a poor candidate for a fine-grained average-case hard problem, as Kane and Williams~\cite{KaneW19} showed that for any fixed \(p \in (0,1)\), $k$-OV instances of i.i.d.~$p$-biased entries can be solved in \(O(n^{k-\varepsilon_p})\) time for some \(\varepsilon_p > 0\) by AC0 circuits.
However, such instances are only nontrivial for \(d = \Theta(\log n)\),\footnote{For larger (resp.\ smaller) $d$, almost all instances will be \no-instances (resp.\ \yes-instances).} a parameter setting which can be anyway solved in time $O(n^{k - \varepsilon})$, even in the worst case, using the algorithm of Chan and Williams~\cite{ChanW21}.
To obtain a candidate hard distribution based on i.i.d.\ entries, we therefore choose to sample the entries as $1$ with subconstant probability $p(n) = o(1)$, which leads to nontrivial instances in the superlogarithmic dimension regime $d = \omega(\log n)$. In
Section~\ref{section:attack}
we present another simple argument why a logarithmic dimension is not sufficient, further justifying our choice.

\paragraph{The $(k-1)$-wise independence.}
Our planting procedure has the following notable property: any subset of $k-1$ (or less) out of the $k$ vectors that form the planted solution has the marginal distribution identical to that of $k-1$ independently sampled vectors with i.i.d.~$p$-biased random entries. In particular, each individual vector of the solution has the same marginal distribution as any other vector in the instance. This would not be true if we planted $k$ random vectors conditioned on orthogonality (i.e., a type of solution that may appear spontaneously with small probability), because such vectors tend to be sparser than the expectation. This sparsity is what makes the Kane--Williams algorithm~\cite{KaneW19} work, and lack thereof makes our instances immune to that algorithmic idea.\footnote{Note though that the Kane--Williams algorithm does not run in truly subquadratic time in superlogarithmic dimension anyway, so above all it is the high dimension, not the $(k-1)$-wise independence, that makes our distribution immune to all known attacks.}

We note that the $(k-1)$-wise independence property holds ``for free'' in natural distributions for $k$-SUM~\cite{AgrawalSSVV24,DinurKK24} and Zero-$k$-Clique~\cite{LaVigneLW19} because of the natural symmetry of cyclic groups $\mathbb{Z}_m$. However, it is a priori unclear how to get it for $k$-OV.

Finally, in Theorem~\ref{thm:unique}, we argue that our distribution is \emph{the unique} distribution over $k$-OV instances that has this property, explaining the title of this paper.

\paragraph{Search-to-decision reductions.}
To demonstrate the usefulness of the $(k-1)$-wise independence property, we give a fine-grained average-case search-to-decision reduction for our conjectured hard $k$-OV distribution. Actually, we give two such reductions. The first one, in Section~\ref{section:search2decision}, is very simple, but it introduces an $O(\log n)$ overhead in the failure probability, so it is relevant only if the decision algorithm succeeds with probability higher than $1-\frac{1}{\log n}$. The other reduction, in Section~\ref{section:search2decision2}, looses only a constant factor in the failure probability. Even though we present both reductions specifically for $k$-OV, they are likely to generalize to any planted problem with the $(k-1)$-wise independence property.

\paragraph{Planting multiple solutions.}
We also argue that $(k-1)$-wise independence allow planting more than one solution in a single instance, which we believe might be useful for building cryptographic primitives.

\subsection{Technical overview}

\paragraph{Planting.}
How do we generate $k$ orthogonal vectors such that any $k-1$ of them look innocent? First of all, we can focus on generating a single coordinate, and then repeat the process independently for each of the $d$ coordinates. Consider the joint distribution of $k$ i.i.d.~$p$-biased random bits. We need to modify it to set the probability of $k$ ones to $0$. If we just do it, and scale up the remaining probabilities accordingly, the probability of $k-1$ ones turns out wrong. After we fix that, the probability of $k-2$ ones is off, and so on, in a manner similar to the inclusion-exclusion principle. By doing this mental exercise we end up with a formula for the joint distribution of $k$ bits in a single coordinate of the $k$ vectors to be planted. How do we actually sample from this distribution? Since it has the $(k-1)$-wise independence property, the following approach must work: First sample $k-1$ i.i.d.~$p$-biased bits, and then sample the $k$-th bit with probability depending on the number of ones among the first $k-1$ bits. In Section~\ref{section:planted_dist} we show how to set this last probability exactly.

\paragraph{Search-to-decision reductions.}
Both our reductions are based on the same basic idea: In order to find the planted solution, we replace some of the vectors in the input instance with newly sampled vectors with i.i.d.~$p$-biased entries and run the decision algorithm on such a modified instance. If at least one of the planted vectors got resampled, the resulting instance has the same distribution as if no planting occurred (thanks to the $(k-1)$-wise independence), and so the decision algorithm returns \no{} with good probability. Otherwise the planted solution is still there and the decision algorithm likely says \yes.

Our first reduction (see Section~\ref{section:search2decision}) applies this idea to perform a binary search. It introduces a factor of $k \log n$ overhead in the running time and also in the failure probability, because we need to take a union bound over all invocations of the decision algorithm returning correct answers.

Our second reduction (see Section~\ref{section:search2decision2}) is an adaptation of a search-to-decision reduction for $k$-SUM due to Agrawal et al.~\cite{AgrawalSSVV24}. In short, the reduction repeatedly resamples a random subset of vectors, runs the decision algorithm, and keeps track for each of the original vectors, how many times the decision algorithm returned \yes{} when this vector was not resampled. Statistically, this count should be larger for vectors in the planted solution. A careful probabilistic analysis shows that this is indeed the case.

\subsection{Open problems}

\paragraph{Hardness self-amplification.}
Could a black-box method increase the success probability of algorithms solving (the search or decision variants of) the planted $k$-OV problem, at a small cost in the running time?
If so, the lack of algorithms with high success probability for planted $k$-OV would then suggest that no algorithm can solve the problem even with just a small success probability -- a property desirable, e.g., from the point of view of potential cryptographic applications.

Such hardness self-amplification was recently shown for (both the search and decision variants of) the planted clique problem by Hirahara and Shimizu~\cite{HS23}. In the world of fine-grained complexity, Agrawal et al.~\cite{AgrawalSSVV24} showed hardness self-amplification for the planted $k$-SUM search problem and closely related problems. Hardness self-amplification for planted $k$-OV remains an open problem.

Because of the overhead in the failure probability induced by both of our search-to-decision reductions, hardness self-amplification for the decision variant of the planted $k$-OV problem in particular would mean that the hardness of the search problem could be based on a weak conjecture about the hardness of the decision problem, such as Conjecture~\ref{conjecture:planted-decision-kov}.

\paragraph{Fine-grained asymmetric cryptography.}
A key goal of fine-grained cryptography is to devise an advanced asymmetric cryptography scheme -- such as public key encryption -- whose security is based on hardness of a well understood problem from fine-grained complexity. So far the closest to this goal seems to be the key exchange protocol due to LaVigne, Lincoln, and Vassilevska Williams~\cite{LaVigneLW19}, which is based on hardness of the planted Zero-$k$-Clique problem. Despite being based on a parameterized problem (that allows for arbitrary polynomial $n^{k-o(1)}$-hardness by simply choosing a large enough $k$), the protocol offers only quadratic security, i.e., breaking the encryption takes only quadratically more than it takes to encrypt and decrypt a message. This limitation seems inherent to the protocol because it is based on a similar idea as Merkle puzzles~\cite{Merkle78}.

It is an open problem if fine-grained cryptography with superquadratic security is possible. We believe that $k$-OV could be a good hard problem for that purpose, because of a different structure, which addition-based problems, like $k$-SUM and Zero-$k$-Clique, are lacking. 

We remark that the key exchange protocol of LaVigne, Lincoln, and Vassilevska Williams~\cite{LaVigneLW19} can be adapted to work with our planted $k$-OV instead of the planted Zero-$k$-Clique problem, but naturally the protocol's security remains quadratic. One needs new techniques to break the quadratic barrier.

In recent work, Alman, Huang, and Yeo~\cite{AlmanHY25} show that if one-way functions do not exist then average-case hardness fine-grained assumptions on planted $k$-SUM and Zero-$k$-Clique are false for sufficiently large constant $k$. It might be possible to generalize their results to $k$-OV. However, a construction of public-key encryption from planted $k$-OV would be interesting even in a world where one-way functions do exist, as they are not known to imply superquadratic-gap public-key encryption~\cite{ImpagliazzoR89}.

\paragraph{Faster algorithms for average-case OV.}
Algorithms for random OV instances seem to be underexplored. Up until recently~\cite{AlmanAZ25} it was not known if the average-case OV admits even a subpolynomial improvement compared to the worst case. With this paper we hope to inspire more research in this direction. We would even be happy to see our conjecture refuted.

A natural starting point for such an attack is the recent Alman--Andoni--Zhang algorithm~\cite{AlmanAZ25} for random OV. It works in time $n^{2-\Omega(\log \log c / \log c)}$ for dimension $d=c \log n$, so it is not truly subquadratic for $d=\omega(\log n)$. However, in their setting, the hardest case is when the probability $p$ of a one entry is chosen to make the \emph{expected number of orthogonal pairs} a constant, while in our setting we use a higher value of $p$ to lower the expectation to inverse polynomial -- which means that in our setting the orthogonal pair ``stands out more''. It might seem plausible to adjust internal parameters of the algorithm (in particular, the so-called \emph{group size}) to better exploit our setting. However, under closer inspection, it turns out that the key quantity in the analysis of the algorithm is the \emph{expected inner product} of two random vectors, equal to $p^2d$, which happens to be $\Theta(\log n)$ in both settings. Refuting our conjecture likely requires a new technical development beyond such an adjustment.

Finally, let us point out a related problem: the planted approximate maximum inner product problem, often referred to as the \emph{light bulb} problem. Unlike planted OV, it is known to admit truly subquadratic $O(n^{1.582})$-time algorithms~\cite{Valiant15,KarppaKK18,KarppaKKC20,Alman19}. This is in contrast with the worst-case complexities of the two problems, which are known to be equivalent under fine-grained reductions~\cite{ChenW19,SM19}.

\paragraph{A worst-case to average-case reduction.}
In the opposite direction than the previous open problem, one could try to show that our conjectured hardness of the planted $k$-OV problem is implied by one of well-studied worst-case hypotheses in fine-grained complexity, e.g., SETH. This would require a worst-case to average-case reduction. So far, in fine-grained complexity, such reductions are only known for algebraic or counting problems~\cite{BRSV17, GoldreichR18, Boix-AdseraBB19, DalirrooyfardLW20, DalirrooyfardLS25}, but not for decision nor search problems like ours.

\section{The model distribution}
Fix $k \geq 2$, and let $d=\alpha(n)\log n$ for $\alpha(n) = \omega(1)$. We define the family of \emph{model distributions} $\Model$ that generate $k$ matrices $U_1,...,U_k \in \{0,1\}^{n\times d}$ where all entries are i.i.d.\ $p$-biased bits with probability
\[
  p = \left( 1 - 2^{-\frac{2k}{\alpha(n)}} \right)^\frac{1}{k} \approx \left(\frac{2k\ln(2)}{\alpha(n)}\right)^{\frac{1}{k}}.
\]
As will become apparent later, for the planting algorithm to work it is crucial that $p \leq 1/2$, but thanks to $\alpha(n) = \omega(1)$ this holds for large enough $n$.

We show that the model distribution indeed generates \no-solutions with good probability.

\begin{lemma}
A $k$-OV instance sampled from the model distribution $\Model$ is a \no-instance with probability at least $1-\frac{1}{n^k}$.
\end{lemma}

\begin{proof}
For a $k$-OV instance $\OV = (U_1,\ldots,U_k) \sim \Model$, a fixed combination of vectors $u_1, \ldots, u_k$ (where $u_\ell \in U_\ell$) is orthogonal iff, for each coordinate $j \in [d]$, not all of the $k$ vectors feature a one in that coordinate. Since $k$ i.i.d.\ $p$-biased bits are all ones with probability $p^k$, the probability that $u_1,\ldots,u_k$ are orthogonal (determined by the all-ones event not occurring in any of the $d$ coordinates) is:
\[
  \Pr[u_1,...,u_k \text{ are orthogonal}] = \left(1 - p^k\right)^d = \left(2^{-\frac{2k}{\alpha(n)}}\right)^{\alpha(n)\log(n)} = n^{-2k}.
\]

By linearity of expectation, the expected value for the number of solutions among all $n^k$ possible combinations of $k$ vectors, denoted by $c(\OV)$, is
\[
  \mathbb{E}[c(\OV)] = \sum_{\substack{u_i \in U_i\\(1 \leq i \leq k)}} \Pr[u_1,\ldots,u_k \text{ are orthogonal}] = n^k \cdot n^{-2k} = \frac{1}{n^k}.
\]
By Markov's inequality, this is also a bound on the probability of \emph{any} solution occurring, i.e., $\Pr[c(\OV) \geq 1] \leq \mathbb{E}[c(\OV)] = \frac{1}{n^k}$.
Therefore, an instance sampled from $\Model$ is a \no-instance with probability at least $1-\frac{1}{n^k}$.
\end{proof}

We remark that one can make the probability of sampling a \no-instance arbitrarily high. Indeed, in order to get the probability $1-\frac{1}{n^c}$ it suffices to replace $2k$ with $k+c$ in the formula for the probability parameter $p$. However, having in mind the cryptographic motivation, $\frac{1}{n^k}$ seems to be a reasonable default choice for the failure probability of the sampler, because with the same probability the attacker can just guess the solution.

\section{The planted distribution}
\label{section:planted_dist}

To plant a solution at locations $s_1,\ldots,s_k \in [n]$ in an instance $\OV$ sampled from $\Model$, we apply the following randomized algorithm.

\paragraph{$\Plant(\OV,s_1,\ldots, s_k)$:}
\begin{enumerate}
  \item For each coordinate $1 \leq j \leq d$:
    \begin{enumerate}
    \item Let $m$ be the number of ones among $U_{1,s_1}[j],\ldots,U_{k,s_k}[j]$.
    \item If $k - m$ is even, flip $U_{k,s_k}[j]$ with probability $\left( \frac{p}{1-p} \right)^{k-m}$. (Here we need $p \leq 1/2$.)
      \end{enumerate}
  \item Return $\OV$.
\end{enumerate}
We justify this way of planting in Section~\ref{section:k-1-wise-independence}.
For now, observe that if all vectors $U_{1,s_1},\ldots,U_{k,s_k}$ feature a one at coordinate $j$, we have $m = k$ and $\Plant$ flips the final bit $U_{k,s_k}[j]$ to a zero with probability
\[
\left( \frac{p}{1-p} \right)^{k-m} = \left( \frac{p}{1-p} \right)^0 = 1.
\]
On the other hand, if the last coordinate is the single zero alongside $m = k - 1$ ones, then $k - m = 1$ is odd and $\Plant$ will never break orthogonality by flipping the last bit to a one. 
Thus, $\Plant(\OV, s_1,\ldots,s_k)$ outputs a \yes-instance of $k$-OV with a solution at $s_1,\ldots,s_k$. We call the $k$ vectors at these positions the \emph{planted vectors}.

We sample \yes-instances of $k$-OV by planting a solution in an instance $\OV \sim \Model$ at locations $s_1,\ldots,s_k$ chosen uniformly at random.

\paragraph{Distribution $\Planted$:}
\begin{enumerate}
  \item Sample $\OV$ from $\Model$.
  \item Sample $(s_1,\ldots,s_k)$ uniformly at random from $[n]^k$.
  \item Return $\Plant(\OV, s_1,\ldots,s_k)$.
\end{enumerate}

The above observation about $\Plant$ immediately yields the following.
\begin{lemma}
A $k$-OV instance sampled from the planted distribution $\Planted$ is a \yes-instance with probability $1$.
\end{lemma}

\section{The \texorpdfstring{$(k-1)$}{(k-1)}-wise independence of planted vectors}\label{section:k-1-wise-independence}
Our method of planting orthogonal vectors arises from the idea that for any planted problem, any proper ``piece'' of the planted solution should be indistinguishable from any comparable piece of the instance as a whole, conditioned on the latter still being consistent with being a part of a solution itself.

For example, in the case of planting a $k$-clique in a graph \(G\) this requirement is trivial. Indeed, the projection of the clique onto a smaller subset of $k' < k$ vertices yields a $k'$-clique, which are exactly those subgraphs of $G$ of size $k'$ which could feasibly belong to a solution.

In contrast to the previous example, in the case of $k$-SUM, any set of $k - 1$ elements $x_1, \dots, x_{k-1}$ in an instance could feasibly be part of a solution, as one can always construct a $k-$th number $x_k$ such that $\sum_{i=1}^k x_i = 0$.
Thus, by the principle we described, to plant a solution in an instance with i.i.d.\ uniformly random elements, the marginal distribution of the distribution of planted solutions $(x_1,\dots,x_k)$ given by any projection to $k-1$ elements should itself be uniformly random.
This holds true in the case of the planted $k$-SUM~\cite{AgrawalSSVV24}, where the planted solution is distributed uniformly over the set of all $k$-tuples that form valid $k$-SUM solutions.
The case of planted Zero-$k$-Clique~\cite{LaVigneLW19} is analogous.
For both of these problems, planting by inserting $k$ elements drawn from the model distribution conditioned on them forming a solution yields a distribution that follows the described principle.

This is different from the $k$-OV problem with a model distribution of i.i.d.\ vector entries.
Here, as with $k$-SUM and Zero-$k$-Clique, any set of $k - 1$ elements (in this case vectors) could form a solution to $k$-OV.
All that is needed is for the last vector to feature a zero in all those coordinates where the other $k-1$ all were one.
However, sparse vectors are far more likely to be part of a solutions than dense ones.
Therefore, conditioning $k$ i.i.d.\ $p$-biased vectors on being orthogonal yields a distribution which does not follow our principle: projecting onto any subset of $k' < k$ vectors results in vectors that are on average sparser than (and thus different from) $k'$ i.i.d.\ $p$-biased vectors.
As we will show now, our method of planting \emph{does} satisfy this principle: Any subset of $k-1$ planted vectors are independent and identically distributed $p$-biased vectors.

Let $M \sim \Model$ and $\OV = \Plant(M, s_1,\ldots,s_k)$.
Recall that both sampling from the model distribution $\Model$ and the planting by $\Plant$ are independent and identical for each coordinate $j \in [d]$. Hence, all $k$-bit sequences $x = (U_{1,s_1}[j], U_{2,s_2}[j], \ldots,  U_{k,s_k}[j]) \in \{0,1\}^k$, for all $j \in [d]$, are independent and identically distributed, according to a distribution whose probability density function we denote by $\Pk : \{0,1\}^k \rightarrow \mathbb{R}$.
\begin{lemma}\label{lem:marginal}
Let $x \in \{0,1\}^k$ and let $m$ be the number of ones in $x$. Then
\[
\Pk(x) = p^m(1-p)^{k-m} - (-1)^{k-m}p^k.
\]
\end{lemma}
\begin{proof}
Fix a coordinate $j \in [d]$. Let $X = (M_{1,s_1}[j], M_{2,s_2}[j], \ldots, M_{k,s_k}[j])$ be the random variable denoting the entries of the $j$-th coordinate among the vectors at locations $s_1, \dots, s_k$ before planting.
We proceed by case distinction.

\proofsubparagraph{Case 1.} If $m - k$ is even, the probability of $x$ occurring in the given coordinate $j \in [d]$ of the planted solution is given by
\begin{align*}
\Pk(x)
={}& \Pr\left[ X = x \text{ and } \Plant \text{ does not flip the final bit} \right]\\
={}& p^m(1-p)^{k-m} \cdot \left( 1 - \left( \frac{p}{1-p} \right)^{k-m} \right)\\
={}& p^m(1-p)^{k-m} - 1 \cdot p^k\\
={}& p^m(1-p)^{k-m} - (-1)^{k-m} p^k.
\end{align*}

\proofsubparagraph{Case 2a.} If $m - k$ is odd and $x = y1$ for some $y \in \{0,1\}^{k-1}$, then $x = y1$ may occur either directly in the model instance, or by $y0$ (for which $m - k$ is even) occurring in the model instance and $\Plant$ flipping the final bit:
\begin{align*}
\Pk(x)
={}& \Pr\left[ X = y1 \right] + \Pr\left[ X = y0 \text{ and } \Plant \text{ flips the final bit} \right]\\
={}& p^m(1-p)^{k-m} + p^{m-1}(1-p)^{k-(m-1)} \cdot \left( \frac{p}{1-p} \right)^{k-(m-1)}\\
={}& p^m(1-p)^{k-m} + p^k\\
={}& p^m(1-p)^{k-m} - (-1)^{k-m}p^k.
\end{align*}

\proofsubparagraph{Case 2b.} Similarly, if $m - k$ is odd and $x = y0$ for some $y \in \{0,1\}^{k-1}$, then $x = y0$ can occur either directly in the model instance or by $\Plant$ flipping the final bit of the sequence $y1$:
\begin{align*}
\Pk(x)
={}& \Pr\left[ X = y0 \right] + \Pr\left[ X = y1 \text{ and } \Plant \text{ flips the final bit} \right]\\
={}& p^m(1-p)^{k-m} + p^{m+1}(1-p)^{k-(m+1)} \cdot \left( \frac{p}{1-p} \right)^{k-(m+1)}\\
={}& p^m(1-p)^{k-m} + p^k\\
={}& p^m(1-p)^{k-m} - (-1)^{m-k} p^k. \qedhere
\end{align*}
\end{proof}

\begin{remark}\label{rem:symmetry}
Despite $\Plant$ acting only on the last collection $U_k$, Lemma~\ref{lem:marginal} implies that the resulting distribution $\Planted$ is invariant under permutation of the sequence of the $k$ collections $U_1,\ldots,U_k$.
\end{remark}

Having $\Pk$ as the distribution of planted vectors, rather than, e.g., the $k$-vector joint model distribution conditioned on orthogonality, ensures \emph{$(k - 1)$-wise independence} among the planted vectors.
I.e., the projection of $k$ planted vectors onto any subset of size $k' < k$ is identically distributed to $k'$ vectors from the model distribution.

\begin{lemma}[$(k-1)$-wise independence]\label{lemma:k-1-independence}
Marginalizing any one of the $k$ bits of $\Pk$ yields $k-1$ independent $p$-biased bits.
\end{lemma}

\begin{proof}
By Remark~\ref{rem:symmetry} we may assume w.l.o.g.\ that the last bit is the one marginalized out.
The lemma then follows from the definition of $\Plant$, as the first $k-1$ entries of any coordinate in the planted vectors are unchanged from the model instance, and are therefore independent $p$-biased bits.
\end{proof}

This property is useful in bounding the probability of a planted instance containing a solution besides the planted one.
\begin{lemma}
A $k$-OV instance sampled from the planted distribution $\Planted$ has more than one solution with probability less than $\frac{1}{n^k}$.
\end{lemma}
\begin{proof}
While the $k$ vectors at positions $s_1, \dots, s_n$ are guaranteed to form a solution, by $(k-1)$-wise independence, \emph{all} combinations of $0 \leq k' < k$ of these vectors and $k - k'$ non-planted vectors form a set of $k$ independent $p$-biased vectors which is therefore a solution to the $k$-OV problem with probability $(1-p^k)^d = \frac{1}{n^{2k}}$.
By linearity of expectation,
\[
  \E[c(\OV)] = 1 + (1 - p^k)^d \cdot (n^k - 1) < 1 + \frac{1}{n^k},
\]
and the claim follows from Markov's inequality.
\end{proof}

\subsection{Uniqueness}
Our way of planting is unique in the following sense.
\begin{theorem}\label{thm:unique}
Let $Q : \{0,1\}^k \to \mathbb{R}$ be a probability distribution such that $Q(1^k) = 0$ and that marginalizing any one of the $k$ bits yields $k-1$ independent $p$-biased bits.
Then $Q = \Pk$.
\end{theorem}
\begin{proof}
We show that $Q(x) = \Pk(x)$ for all $x \in \{0,1\}^k$.
Let $m$ denote the number of ones in $x$.
We proceed by induction over $k - m$, i.e., the number of \emph{zeros} in $x$.
\proofsubparagraph{Base case: $k - m = 0$.}
Then $m = k$ and $x = 1^k$.
Thus $Q(x) = Q(1^k) = 0 = \Pk(x)$.

\proofsubparagraph{Inductive case: $k - m > 0$.}
We assume w.l.o.g.\ that the $k - m$ zeros are the \emph{last} bits of $x$, i.e., $x = 1^m0^{k-m}$.
Marginalizing the final $k-m > 0$ bits of $Q$ yields $k - (k - m) = m$ independent $p$-biased bits, whereby the probability of all $m$ remaining bits being ones is
\[
p^m = \sum_{y \in \{0,1\}^{k-m}} Q(1^m y) = Q(\underbrace{1^m0^{m-k}}_{= x}) + \sum_{\substack{y \in \{0,1\}^{k-m}\\y \not= 0^{k-m}}} Q(1^m y).
\]
Thereby, 
\begin{align*}
Q(x) &= p^m - \sum_{\substack{y \in \{0,1\}^{k-m}\\y \not= 0^{k-m}}} Q(1^m y)\\
&= p^m - \sum_{\substack{y \in \{0,1\}^{k-m}\\y \not= 0^{k-m}}} \Pk(1^m y) \tag*{(by the induction hypothesis)}\\
&= p^m + \Pk(x) - \sum_{y \in \{0,1\}^{k-m}} \Pk(1^m y)\\
\intertext{where the sum term is merely the probability of $1^m$ in the marginal distribution of $\Pk$, which by Lemma~\ref{lemma:k-1-independence} in turn consists of $m$ independent $p$-biased bits. Hence,}
&= p^m + \Pk(x) - p^m = \Pk(x).
\end{align*}
\end{proof}

\section{Conjectured hard problems}

In this section we formally define the problems that we conjecture to require $n^{k-o(1)}$ time.

\begin{definition}[Solving planted \textbf{decision} $k$-OV]
  Let $\mathcal{A}$ be an algorithm that given a $k$-OV instance $\OV$ outputs either 0 or 1.
  For $\alpha(n) = \Omega(1)$, we say $\mathcal{A}$ \emph{solves} the \emph{decision} $k\text{-OV}^\alpha$ problem with success probability $\delta(n)$, if for both $b \in \{0,1\}$ and large enough $n$,
  \[
    \Pr_{\OV \sim k\text{-OV}_b^\alpha(n)} [\mathcal{A}(\OV) = b] \geq \delta(n),
  \]
  where randomness is taken over both the instance $\OV$ and the random coins used by $\mathcal{A}$.
\end{definition}

Similarly, we define a notion of recovering a solution from a planted instance.
\begin{definition}[Solving planted \textbf{search} $k$-OV]
  Let $\mathcal{A}$ be an algorithm that given a $k$-OV instance $\OV$ outputs a tuple $(s_1, ..., s_k) \in \{1,...,n\}^k$.
  For a given $\alpha(n) = \Omega(1)$, we say $\mathcal{A}$ \emph{solves} the \emph{planted search $k\text{-OV}^\alpha$ problem} with success probability $\delta(n)$ if for large enough $n$,
  \[
    \Pr_{\substack{\OV \sim \Planted\\(s_1,\ldots,s_k) \gets \mathcal{A}(\OV) }} [U_{1,s_1},...,U_{k,s_k} \text{ are orthogonal}] \geq \delta(n),
  \]
  where randomness is taken over both the instance $\OV$ and the random coins used by $\mathcal{A}$.
\end{definition}

Now we are ready to formally state our main conjecture.

\begin{conjecture}\label{conjecture:planted-decision-kov}
  For any $\alpha(n) = \omega(1)$ and $\varepsilon > 0$, there exists no algorithm $\mathcal{A}$ that solves the planted decision $k$-OV problem with any constant success probability $\delta > \frac{1}{2}$ in time $O(n^{k-\varepsilon})$.
\end{conjecture}

\section{Search-to-decision reduction via binary search}\label{section:search2decision}
\newcommand{\Ad}{\mathcal{A}^\mathrm{decide}}
\newcommand{\As}{\mathcal{A}^{\mathrm{search}}}
We reduce the search problem of finding the planted solution to the decision problem of determining whether an instance contains a planted solution.
This means that given a decision algorithm that can correctly distinguish whether an instance was sampled from the model or planted distribution with sufficient probability, one can recover the planted secret through this reduction.
The reduction introduces a factor $O(\log n)$ increase in both the running time and error probability of the algorithm.

The idea is to find each planted vector using something akin to binary search on each collection $U_i$.
We can split $U_i$ into two partitions of roughly equal size and run the decision algorithm twice, on instances where one of the two partitions is first replaced by newly sampled $p$-biased vectors.
The vector planted in $U_i$ is guaranteed to be replaced in one of these cases, and by $(k-1)$-wise independence the resulting instance follows the model distribution.
The search space is thus cut in half and we can recurse on this smaller search space to eventually find the planted vector.

\begin{theorem}[Search-to-decision reduction]\label{theorem:search-to-decision}
  Let $\alpha(n) = \polylog(n)$ and let $\Ad$ be an algorithm that solves the planted decision $k$-OV problem with success probability $1 - \delta(n)$ in time $T(n)$.
  Then there exists an algorithm $\As$ that solves the planted search $k$-OV problem with success probability at least
  $1 - k\lceil\log n\rceil \cdot \delta(n)$ in expected time $\Olog(T(n) + n)$.
\end{theorem}

\begin{proof}
  Consider an instance $\inst = (U_1, \dots, U_k) \sim \Planted$.
  First let us focus only on recovering the location $i \in [n]$ of the planted vector in the first collection $U_1$.
  The reduction begins with the ``full'' search space $S := [n]$, and narrows it down by half in each iteration, so that the desired $i$ is recovered after $\lceil \log n \rceil$ iterations.

  At each iteration, the current search space $S$ is arbitrarily partitioned in two sets of equal size (up to one vector when $|S|$ is odd).
  The decision algorithm $\Ad$ is then executed on two new instances, where the respective sets of vectors in $U_1$ are replaced with newly sampled $p$-biased vectors.

  By the $(k-1)$-wise independence, if the vector belonging to the solution is replaced, all vectors are independently and identically distributed $p$-biased vectors, i.e., the instance is distributed according to $\Model$.
  On the other hand, if the solution survives resampling, the instance remains distributed according to $\Planted$.
  Therefore, the output of $\Ad$ is used to decide which of the two  partition blocks should be assumed as the new search space.

  The reduction is correct if $\Ad$ decides correctly at every iteration.
  Of course, $\Ad$ might fail with probability $\delta(n)$. By a union bound over all $\lceil \log n \rceil$ invocations of $\Ad$ this happens with probability at most $\lceil \log n \rceil \cdot \delta(n)$.
  Thereby $\As$ recovers the location $i$ of the first planted vector with success probability at least $1 - \lceil \log n \rceil \cdot \delta(n)$.

  As for the runtime, $\Ad$ with runtime $T(n)$ is invoked $O(\log n)$ times, and across all iterations $n-1$ vectors are resampled in total.
Since a single $p$-biased bit can be sampled in expected time $O(- \log p) = O(\log \alpha(n)) = O(\log \log n)$, sampling a $d$-dimensional vector takes $\polylog(n)$ time in expectation.
  Therefore, recovering the location of the first planted vector takes time $\Olog(T(n) + n)$.

  The same process is repeated another $k - 1$ times to recover the locations of the planted vectors among $U_2,...,U_k$.
  As $k$ is constant, this does not increase the running time asymptotically but the success probability drops to $1 - k\lceil\log n\rceil \cdot \delta(n)$.
\end{proof}

\section{Search-to-decision reduction via counters}
\label{section:search2decision2}
\newcommand{\Mix}{\mathsf{Mix}}
We present a second search-to-decision reduction, adapted from that of Agrawal et al.~\cite{AgrawalSSVV24} for planted $k$-SUM.
As in the method in Section~\ref{section:search2decision}, we use the fact that an algorithm $\Ad$ for the decision $k$-OV problem, when given a planted $k$-OV instance with some of the vectors resampled, correctly detects whether any of the planted vectors were among the resampled vectors.
However, instead of iteratively narrowing a pool of candidate vectors, we iterate this process on the entire instance, and, for each vector $u$, we keep count of the number of iterations in which $u$ survived and $\Ad$'s output was 1.
After $O(\log(n))$ iterations we output the vectors with the highest counts among each of the $k$ collections, which, as we show, coincides with the planted solution (with good probability).

\begin{theorem}[Search-to-decision reduction]\label{theorem:k-sum-style-s2d}
  For any $\alpha(n) = \polylog(n)$, if there exists an algorithm that solves the planted decision $k$-OV problem with success probability at least $1 - \delta(n)$ in time $T(n)$,
  then there exists an algorithm that solves the planted search $k$-OV problem with success probability at least $1 - 13k \cdot \delta(n) - \frac{1}{n^k}$ in expected time
  $O\left(\big(T(n) + n \polylog(n)\big) \log \frac{n}{\delta(n)}\right)$.
\end{theorem}

In more detail, let $\Mix$ be the following randomized algorithm, which takes a $k$-OV instance $\inst$ and resamples some of the vectors:

\paragraph{Algorithm $\Mix(\inst)$:}
\begin{enumerate}
  \item For each $\ell \in [k]$ and $i \in [n]$:
    \begin{enumerate}
    \item With probability $1 - 2^{-\frac{1}{k}}$, replace $U_{\ell,i}$ by a newly-sampled $p$-biased vector
      \end{enumerate}
  \item Output $\inst$
\end{enumerate}

For $\ell \in [k]$, let $R_\ell \subseteq [n]$ indicate the indices of the vectors of $U_\ell$ which are replaced by $\Mix$.
For a vector $u$ in $\inst$ and a given execution of $\Mix$, we say $u$ \emph{survives} if $\Mix$ does not replace $u$.
We say $\mathbf{s} = (s_1, \dots, s_k)$ \emph{survives} if $U_{\ell, s_\ell}$ survives for each $\ell \in [k]$, i.e., $\forall_{\ell \in [k]} s_\ell \notin R_\ell$.

Now, let $B(\mathbf{s})$ be the binary random variable indicating whether $\mathbf{s}$ survives.
Our chosen probability for $\Mix$ to resample a vector yields the following.
\begin{lemma}
For a $k$-OV instance $\inst$ and any $\mathbf{s} \in [n]^k$, $\mathbf{s}$ survives with probability one half, i.e.,
  $\Pr_{\Mix}[B(\mathbf{s}) = 1] = \frac{1}{2}$.
\end{lemma}
\begin{proof}
Each vector is independently picked to be replaced with probability $1 - 2^{-\frac{1}{k}}$.
The chance of all $k$ vectors surviving is
\[
  \Pr_{\Mix}[B(\mathbf{s}) = 1] = \left(1 - \left( 1 - 2^{-\frac{1}{k}} \right) \right)^k = 2^{-\frac{1}{k} \cdot k} = \frac{1}{2}.
\qedhere
\]
\end{proof}

As laid out before, our search algorithm repeatedly executes $\Mix$ and then the given decision algorithm $\Ad$. The hope is that the output of the latter correlates with the survival of the planted solution.
We therefore also keep track of which vectors survive whenever $\Ad$ believes there is still a planted solution in the instance output by $\Mix$.

\paragraph{Algorithm $\As(\inst)$:}
\begin{enumerate}
  \item Initialize a counter $C_{\ell,i} := 0$ for each $\ell \in [k]$ and $i \in [n]$.
  \item Repeat $m = \Theta(\log n)$ times:
    \begin{enumerate}
    \item $\inst[V] \gets \Mix(\inst)$
    \item $b := \Ad(\inst[V])$
    \item If $b = 1$:
    \begin{enumerate}
    \item Set $C_{\ell,i} := C_{\ell,i} + 1$ for every $U_{\ell,i}$ that was not replaced by $\Mix$
    \end{enumerate}
      \end{enumerate}
  \item Set $s_\ell := \argmax_{i \in [n]} C_{\ell,i}$ for each $\ell \in [k]$
  \item Output $\mathbf{s} = (s_1, \dots, s_k)$
\end{enumerate}

This works well for instances $\inst$ where $\Ad$ is good at detecting whether a particular solution survives.
To capture this notion, we say an instance $\inst$ is \emph{good}, if it has only one solution, at some location $\mathbf{s}$, \emph{and} the output of $\Ad$ indicates whether $\mathbf{s}$ survives except for a small constant probability, i.e.,
\[
\Pr[\Ad(\Mix(\inst)) \not= B(\mathbf{s})] < \frac{1}{12k},
\]
where the probability is taken over the internal randomness used by $\Mix$.

In the following, let $\mathsf{Sample}$ be a randomized algorithm that outputs a planted instance sampled from $\Planted$ as well as the location $\mathbf{s}$ of the planted solution.

\begin{lemma}\label{lemma:most-Us-good}
If $\Ad$ solves the planted decision $k$-OV problem with success probability at least $1 - \delta(n)$, an instance $\inst \sim \Planted$ is \emph{good} except with probability at most $12k \cdot \delta(n) + \frac{1}{n^k}$.
\end{lemma}
\begin{proof}
An instance $\inst \sim \Planted$ contains only a single solution except with probability less than $\frac{1}{n^k}$.
We will now show that for such $\inst$ with only the planted solution, the decision algorithm $\Ad$ correctly detects if this solution survives except with probability $< 12k \cdot \delta(n)$, from which the claim follows by a union bound.

First, consider the following distribution:

\paragraph{Distribution $\Dist{B}$:}
\begin{enumerate}
  \item $(\inst[U], \mathbf{s}) \gets \mathsf{Sample}$
  \item $\inst[V] \gets \Mix(\inst[U])$
  \item Output $(\inst[V], B(\mathbf{s})).$
\end{enumerate}

Observe that, if we condition on $B(\mathbf{s}) = 1$, the planted solution survives $\Mix$, which merely resamples some of the i.i.d.\ $p$-biased vectors in the instance.
Thus, $\inst[V]$ is distributed according to $\Planted$.
On the other hand, if $B(\mathbf{s}) = 0$, at least one of the planted vectors is replaced by a newly-sampled $p$-biased vector.
By $(k-1)$-wise independence, the subset of planted vectors that survive are i.i.d.\ $p$-biased vectors, as are all other vectors in $\inst[V]$.
Hence, conditioning on $B(\mathbf{s}) = 0$ yields the model distribution $\Model$.

Therefore, for the algorithm $\Ad$, which solves the planted decision $k$-OV problem with success probability at least $1 - \delta(n)$, we have
\[
  \Pr_{(\inst[V], B) \gets \Dist{B}}[\Ad(\inst[V]) \not= B] \leq \delta(n).
\]
Now, let $Z(\OV, \mathbf{s})$ be the random variable that, for a given instance $\inst$ with a solution planted at $\mathbf{s}$, denotes the probability of $\Ad(\Mix(U)) \not= B(\mathbf{s})$, where the randomness is taken over the internal coins used by $\Mix$.
Then
\begin{align*}
  \delta(n) &\geq \Pr_{(\inst[V], B) \gets \Dist{B}}[\Ad(\inst[V]) \not= B]\\
  &= \E_{\substack{(\inst, \mathbf{s}) \gets \mathsf{Sample}}} \left[ \Pr_{\Mix}[\Ad(\Mix(\inst)) \not= B(\mathbf{s})] \right]\\
  &= \E_{(\inst, \mathbf{s}) \gets \mathsf{Sample}} \left[ Z(\inst, \mathbf{s}) \right].
\end{align*}
By Markov's inequality,
\[
  \Pr_{(\inst, \mathbf{s}) \gets \mathsf{Sample}} \left[ Z(\inst, \mathbf{s}) \geq \frac{1}{12k} \right]
  \leq 12k \E[Z(\inst, \mathbf{s})] \leq 12 k \cdot \delta(n).
\]

Next, observe that $\mathbf{s}$ is the only solution in the instance $\inst$ output by $\mathsf{Sample}$ with good probability,
\[
  \Pr_{(\inst, \mathbf{s}) \gets \mathsf{Sample}} \left[ \inst  \text{ has a solution besides } \mathbf{s} \right]
  = \Pr_{\inst \gets \Planted} \left[ c(\inst) > 1 \right] < \frac{1}{n^k}.
\]
Thus, by a union bound over this and our result in the first step, we find that an instance $\inst \sim \Planted$ is good except with probability at most $12k \cdot \delta(n) + \frac{1}{n^k}$:
\begin{align*}
  & \Pr_{\inst \gets \Planted} \left[ \inst \text{ is good} \right]\\
  ={}& \Pr_{\inst, \mathbf{s} \gets \mathsf{Sample}} \left[ \mathbf{s} \text{ is the only solution of } \inst \text{ and } Z(\inst, \mathbf{s}) < \frac{1}{12k} \right]\\
  \geq{}& 1 - \Pr_{\inst, \mathbf{s} \gets \mathsf{Sample}} \left[\inst[U] \text{ has a solution besides } \mathbf{s} \right]
             - \Pr_{\inst, \mathbf{s} \gets \mathsf{Sample}} \left[ Z(\inst, \mathbf{s}) \geq \frac{1}{12k} \right]\\
  >{}& 1 - \frac{1}{n^k} - 12k \cdot \delta(n). \qedhere
\end{align*}
\end{proof}

We now show that the search algorithm performs well on this large fraction of good instances.
\begin{lemma}\label{lemma:asearch-works-on-good-Us}
Let $\Ad$ be an algorithm that solves the planted decision $k$-OV problem with success probability $1 - \delta(n)$.
Then $\As$ fails to recover the solution $\mathbf{s}$ of good instances with probability less than $\delta(n)$.
\end{lemma}
\begin{proof}
Let $\inst$ be a good instance with its only solution at $\mathbf{s}$.
After $t$ iterations, we expect the counters for vectors in the solution $\mathbf{s}$ to be the highest. Using the fact that $\frac{1}{\sqrt[k]{2}} < 1 - \frac{1}{2k}$, we find that for non-planted vectors, i.e., where $s_\ell \not = i$,
\begin{align}
  \E[C_{\ell,i}]
  &= t \cdot \Pr_{\Mix} \left[ U_{\ell,i} \text{ survives and } \Ad(\Mix(\inst)) = 1 \right] \tag{$\ast$}\label{eq:exp-non-planted}
  \\ &\leq t \cdot \left( \Pr_{\Mix} \left[ U_{\ell,i} \text{ survives and } B(\mathbf{s}) = 1\right] + \Pr_{\Mix} \left[ \Ad(\Mix(\inst)) \not= B(\mathbf{s}) \right] \right) \notag
  \\ &< t \cdot \left( \frac{1}{\sqrt[k]{2}} \cdot \frac{1}{2} + \frac{1}{12k} \right)
  < t \cdot \left(\left(1 - \frac{1}{2k}\right)\frac{1}{2} + \frac{1}{12k} \right) \notag\\
  &= t \left(\frac{1}{2} - \frac{2}{12k} \right). \notag
       \\[10pt] 
       \intertext{On the other hand, for planted vectors, i.e., where $s_\ell = i$,}
  \E[C_{\ell,i}]
  &\geq t \cdot \Pr_{\Mix} \left[ \mathbf{s} \text{ survives and } \Ad(\Mix(\inst)) = 1 \right] \tag{$\ast\ast$}\label{eq:exp-planted}\\
  &= t \cdot \Pr_{\Mix} \left[ B(\mathbf{s}) = 1 \text{ and } \Ad(\Mix(\inst)) = B(\mathbf{s}) \right] \notag\\
  &\geq t \cdot \left(1 - \Pr_{\Mix} \left[ B(\mathbf{s}) = 0\right] - \Pr_{\Mix} \left[ \Ad(\Mix(\inst)) \not= B(\mathbf{s}) \right] \right) \notag\\
  &> t \cdot \left(1 - \frac{1}{2} - \frac{1}{12k} \right)
  = t \cdot \left(\frac{1}{2} - \frac{1}{12k} \right).\notag
\end{align}
Picking the $k$ highest counters is guaranteed to yield the solution $\mathbf{s}$ if the ranges of the counters of the planted and non-planted vectors do not overlap, that is to say if no counter deviates from its expected value by half the difference of (the bounds on) the two expected values~\eqref{eq:exp-planted} and~\eqref{eq:exp-non-planted}, which we denote by $\Delta$:
\[
  \Delta := \frac{1}{2} \left[t \cdot \left(\frac{1}{2} - \frac{1}{12k} \right)
               - t \cdot \left( \frac{1}{2} - \frac{2}{12k} \right) \right] = t \frac{1}{24k}.
\]

Each counter $C_{\ell,i}$ is the sum of $t$ i.i.d.\ binary random variables.
By a Chernoff bound, a counter $C_{\ell,i}$ for a vector which is not in the planted solution, i.e., $s_\ell \not= i$, exceeds its expected value by $\Delta$ with probability at most
\[
  \Pr \left[ C_{\ell,i} \geq \E[C_{\ell,i}] + \Delta \right]
  < \mathrm{exp}\left( - 2t \left( \frac{\Delta}{t}\right)^2\right)
  = \mathrm{exp}\left( - \frac{2t}{24^2 \cdot k^2}\right).
\]
Similarly, a counter $C_{\ell,i}$ for a vector that \emph{is} part of a planted solution ($s_\ell = i$) falls short of its expected value by $\Delta$ with at most the same probability
\[
  \Pr \left[ C_{\ell,i} \leq \E[C_{\ell,i}] - \Delta \right]
  < \mathrm{exp}\left( - 2t \left( \frac{\Delta}{t}\right)^2\right)
  = \mathrm{exp}\left( - \frac{2t}{24^2 \cdot k^2}\right).
\]
For a choice of $t = \frac{24^2 \cdot k^2}{2} \cdot \log \frac{kn}{\delta(n)} = O\left(\log \frac{n}{\delta(n)}\right)$ iterations, both of these probabilities are at most $\frac{\delta(n)}{kn}$.
If \emph{none} of the $k \cdot n$ counters deviate by $\Delta$, the ranges of counters for vectors at $\mathbf{s}$ and those for vectors not at $\mathbf{s}$ are disjoint and selecting for the highest counters is guaranteed to yield $\mathbf{s}$.
Thus, by a union bound over all $k \cdot n$ counters, we fail to recover $\mathbf{s}$ with probability at most $kn \cdot \frac{\delta(n)}{kn} = \delta(n)$.
\end{proof}

We can now complete the proof of Theorem~\ref{theorem:k-sum-style-s2d}.
\begin{proof}[Proof of Theorem \ref{theorem:k-sum-style-s2d}]
By Lemma~\ref{lemma:most-Us-good}, an instance $\inst \sim \Planted$ is good except with probability at most $12k \cdot \delta(n) + \frac{1}{n^k}$.
By Lemma~\ref{lemma:asearch-works-on-good-Us}, $\As$ is able to recover the solution $\mathbf{s}$ from a good instance except with probability at most $\delta(n)$.
By a union bound against the instance not being good or $\As$ failing on a good instance, $\As$ solves the planted search $k$-OV problem with success probability at least $1 - 13k\delta(n) - \frac{1}{n^k}$.

In each of the
$O\left(\log \frac{n}{\delta(n)}\right)$
iterations, we execute both $\Mix$ and $\Ad$ once and update the counters.
Each execution of $\Mix$ samples, in expectation, $k \cdot n \cdot (1 - 2^{-\frac{1}{k}}) = O(n)$ new vectors, each of which can be sampled in $\polylog(n)$ expected time as explained in the proof of Theorem~\ref{theorem:search-to-decision}.
$\Ad$ runs in time $T(n)$ and updating the counters takes linear time, for the total expected runtime of
$O\left(\big(T(n) + n \polylog(n)\big) \log \frac{n}{\delta(n)}\right)$.
\end{proof}

\section{Planting more solutions}
So far we have only considered planting a single solution in an instance sampled from the model distribution.
We justified our particular planted distribution by the fact it is the only way of sampling solutions (orthogonal combinations of vectors) such that every subset of $k' < k$ vectors is distributed identically to $k'$ independent $p$-biased vectors.
As a consequence, if we have multiple sets of orthogonal vectors sampled this way, any combination of $k$ vectors that are not all part of the same set are distributed identically to $k$ independent $p$-biased vectors.

We propose a family of distribution, called $m$-\emph{planted}-$\Model$, for $m \in [n]$, where $m$ sets of $k$ orthogonal vectors are sampled independently and embedded at $m$ (disjoint) locations in a $\Model$ instance.
This is a generalization of the previous planted distribution which corresponds to the case of $m = 1$.
On the other extreme, when $m = n$, all $k \cdot n$ vectors in an $m$-planted-$\Model$ instance belong to some orthogonal combination and no vector is an independently sampled $p$-biased vector.
The task of finding a single orthogonal combination of vectors is therefore trivially easier in this case than in the case of a $1$-planted $\Model$ instance, as one can simply pick any one vector arbitrarily, and find $k - 1$ vectors from the other collections to form an orthogonal combination through exhaustive search in time $O(n^{k-1}d) = o(n^k)$.
Recovering the entire solution, that is all $n$ many $k$-tuples of locations where the orthogonal combinations were planted, still takes time $\Theta(n^k d)$ when done in the naive way.

In general, a solution to the $m$-planted $\Model$ problem can be encoded as $k$ sets $S_1,...,S_k \subseteq \{1,...,n\}$ of size $m$, indicating which vectors within collection $U_1,...,U_k$ are part of an orthogonal combination, as well as $k-1$ permutations $\pi_2,\dots,\pi_k$ on the sets $S_2, \dots, S_k$ respectively that encode how these vectors match up to form orthogonal combinations.
We call $S_1,...,S_k$ the \emph{support of the solution} and $\pi_2,...,\pi_k$ the \emph{matching}.
There are therefore
\[
  \binom{n}{m}^k \cdot (m!)^{k-1}
\]
possible solutions and the information encoded in an instance amounts to
\begin{align*}
  \log\left( \binom{n}{m}^k \cdot (m!)^{k-1} \right) &= k \log \binom{n}{m} + (k-1) \log (m!)\\
                                                     &= O(m \log n)
\end{align*}

As an observation, in the case of $1$-planted $\Model$, the matching is trivial (as there is only one planted orthogonal combination), and therefore the difficulty in the problem is derived from finding the support.
On the other hand, for $n$-planted $\Model$, the support is trivial (as $S_1 = ... = S_k = \{1,...,n\}$) and the difficulty lies solely in determining the matching.

\section{The downsampling attack}
\label{section:attack}

In this section we give a simpler argument for choosing $d=\omega(\log n)$, which does not rely on the algorithms by Abboud et al.~\cite{AbboudWY15} and Chan and Williams~\cite{ChanW21}. We call it the \emph{downsampling attack}. It makes use of the fact that if we project higher dimensional orthogonal vectors onto a subset of coordinates, the vectors remain orthogonal. Moreover, if we project an entire $k$-OV instance onto less than $\log n$ coordinates, then there are less than $n$ different possible vectors, and hence we can efficiently enumerate all possible solutions. However, non-orthogonal vectors may also become orthogonal under such projection  -- as long as we ensure that such downsampling does not create too many false positives. We show that this approach leads to an $O(n^{k-\varepsilon})$-time algorithm, for $\varepsilon > 0$, when $d=O(\log n)$.

\begin{theorem}
Planted (search) $k$-OV in dimension $d=\alpha(n) \log n$ can be solved in expected time $O\left(n^{k-\frac{k}{1 + \alpha(n)}}\right)$, with the expectation taken over the instance.
\end{theorem}

\begin{proof}
The algorithm begins by looking only at the first $d' := \frac{d}{1+\alpha(n)} = \frac{\alpha(n)}{1+\alpha(n)} \log n < \log n$ coordinates of each vector. It iterates over all possible $k$-tuples of $d'$-dimensional vectors that are orthogonal. There are at most $O(2^{d'k}) = O\big(n^{k-\frac{k}{1 + \alpha(n)}}\big)$ of them. Then, for each such low-dimensional orthogonal tuple, the algorithm lists all original high-dimensional $k$-tuples that project onto it, and checks if any of them is indeed orthogonal when considering all $d$ coordinates.

The total running time of the algorithm is proportional to the number of all possible low-dimensional orthogonal tuples, which we have already bound in the previous paragraph, plus the number of high-dimensional tuples that project onto one of them, which we bound as follows. Let $p_\perp = (1-p^k)^d$ denote the probability that $k$ random vectors are orthogonal. The input distribution\footnote{This is true about our proposed distribution, but one could argue that this argument (and hence also this whole proof) applies to \emph{any} potential planted $k$-OV distribution, as otherwise it would produce in expectation at least one spurious solution.} satisfies $p_\perp \leq 1/n^k$. The probability that $k$ random vectors restricted to the first $d'$ coordinates are orthogonal is
\[(1-p^k)^{d'} = p_\perp^{(d'/d)} = p_\perp^{1/(1+\alpha(n))} \leq (1/n^k)^{1/(1+\alpha(n))} = n^{-(k/(1+\alpha(n)))}.\]
We multiply this probability by $n^k$ to get that the expected number of such $k$-tuples is $n^{k-\frac{k}{1 + \alpha(n)}}$.
\end{proof}

\section*{Acknowledgements.} We are grateful to Andrej Bogdanov, Antoine Joux, Moni Naor, Nicolas Resch, Nikolaj Schwartzbach, and Prashant Vasudevan for insightful discussions.

\bibliographystyle{plainurl}
\bibliography{paper}

\begin{thebibliography}{10}

\bibitem{AbboudWY15}
Amir Abboud, Richard~Ryan Williams, and Huacheng Yu.
\newblock More applications of the polynomial method to algorithm design.
\newblock In {\em Proceedings of the Twenty-Sixth Annual {ACM-SIAM} Symposium
  on Discrete Algorithms, {SODA} 2015}, pages 218--230. {SIAM}, 2015.
\newblock \href {https://doi.org/10.1137/1.9781611973730.17}
  {\path{doi:10.1137/1.9781611973730.17}}.

\bibitem{AgrawalSSVV24}
Shweta Agrawal, Sagnik Saha, Nikolaj~I. Schwartzbach, Akhil Vanukuri, and
  Prashant~Nalini Vasudevan.
\newblock k-{SUM} in the sparse regime: Complexity and applications.
\newblock In {\em Advances in Cryptology -- {CRYPTO} 2024 -- 44th Annual
  International Cryptology Conference}, volume 14921 of {\em Lecture Notes in
  Computer Science}, pages 315--351. Springer, 2024.
\newblock \href {https://doi.org/10.1007/978-3-031-68379-4\_10}
  {\path{doi:10.1007/978-3-031-68379-4\_10}}.

\bibitem{Alman19}
Josh Alman.
\newblock An illuminating algorithm for the light bulb problem.
\newblock In {\em 2nd Symposium on Simplicity in Algorithms, {SOSA} 2019},
  volume~69 of {\em OASIcs}, pages 2:1--2:11. Schloss Dagstuhl --
  Leibniz-Zentrum f{\"{u}}r Informatik, 2019.
\newblock \href {https://doi.org/10.4230/OASICS.SOSA.2019.2}
  {\path{doi:10.4230/OASICS.SOSA.2019.2}}.

\bibitem{AlmanAZ25}
Josh Alman, Alexandr Andoni, and Hengjie Zhang.
\newblock Faster algorithms for average-case orthogonal vectors and closest
  pair problems.
\newblock In {\em 2025 Symposium on Simplicity in Algorithms (SOSA)}, pages
  473--484. SIAM, 2025.
\newblock \href {https://doi.org/10.1137/1.9781611978315.35}
  {\path{doi:10.1137/1.9781611978315.35}}.

\bibitem{AlmanHY25}
Josh Alman, Yizhi Huang, and Kevin Yeo.
\newblock Fine-grained complexity in a world without cryptography.
\newblock In {\em Advances in Cryptology -- {EUROCRYPT} 2025 -- 44th Annual
  International Conference on the Theory and Applications of Cryptographic
  Techniques}, volume 15607 of {\em Lecture Notes in Computer Science}, pages
  375--405. Springer, 2025.
\newblock \href {https://doi.org/10.1007/978-3-031-91098-2\_14}
  {\path{doi:10.1007/978-3-031-91098-2\_14}}.

\bibitem{BRSV17}
Marshall Ball, Alon Rosen, Manuel Sabin, and Prashant~Nalini Vasudevan.
\newblock Average-case fine-grained hardness.
\newblock In {\em Proceedings of the 49th Annual {ACM} {SIGACT} Symposium on
  Theory of Computing, {STOC} 2017}, pages 483--496. {ACM}, 2017.
\newblock \href {https://doi.org/10.1145/3055399.3055466}
  {\path{doi:10.1145/3055399.3055466}}.

\bibitem{Boix-AdseraBB19}
Enric Boix{-}Adser{\`{a}}, Matthew~S. Brennan, and Guy Bresler.
\newblock The average-case complexity of counting cliques in
  {E}rd{\H{o}}s-{R}{\'{e}}nyi hypergraphs.
\newblock In {\em 60th {IEEE} Annual Symposium on Foundations of Computer
  Science, {FOCS} 2019}, pages 1256--1280. {IEEE} Computer Society, 2019.
\newblock \href {https://doi.org/10.1109/FOCS.2019.00078}
  {\path{doi:10.1109/FOCS.2019.00078}}.

\bibitem{ChanW21}
Timothy~M. Chan and R.~Ryan Williams.
\newblock Deterministic {APSP}, orthogonal vectors, and more: Quickly
  derandomizing {R}azborov-{S}molensky.
\newblock {\em {ACM} Trans. Algorithms}, 17(1):2:1--2:14, 2021.
\newblock Announced at {SODA} 2016.
\newblock \href {https://doi.org/10.1145/3402926} {\path{doi:10.1145/3402926}}.

\bibitem{ChenW19}
Lijie Chen and Ryan Williams.
\newblock An equivalence class for orthogonal vectors.
\newblock In {\em Proceedings of the Thirtieth Annual {ACM-SIAM} Symposium on
  Discrete Algorithms, {SODA} 2019}, pages 21--40. {SIAM}, 2019.
\newblock \href {https://doi.org/10.1137/1.9781611975482.2}
  {\path{doi:10.1137/1.9781611975482.2}}.

\bibitem{DalirrooyfardLS25}
Mina Dalirrooyfard, Andrea Lincoln, Barna Saha, and Virginia {Vassilevska
  Williams}.
\newblock Average-case hardness of parity problems: Orthogonal vectors, k-{SUM}
  and more.
\newblock In {\em Proceedings of the 2025 Annual {ACM-SIAM} Symposium on
  Discrete Algorithms, {SODA} 2025}, pages 4613--4643. {SIAM}, 2025.
\newblock \href {https://doi.org/10.1137/1.9781611978322.158}
  {\path{doi:10.1137/1.9781611978322.158}}.

\bibitem{DalirrooyfardLW20}
Mina Dalirrooyfard, Andrea Lincoln, and Virginia {Vassilevska Williams}.
\newblock New techniques for proving fine-grained average-case hardness.
\newblock In {\em 61st {IEEE} Annual Symposium on Foundations of Computer
  Science, {FOCS} 2020}, pages 774--785. {IEEE}, 2020.
\newblock \href {https://doi.org/10.1109/FOCS46700.2020.00077}
  {\path{doi:10.1109/FOCS46700.2020.00077}}.

\bibitem{DinurKK24}
Itai Dinur, Nathan Keller, and Ohad Klein.
\newblock Fine-grained cryptanalysis: Tight conditional bounds for dense
  k-{SUM} and k-{XOR}.
\newblock {\em J. {ACM}}, 71(3):23, 2024.
\newblock Announced at {FOCS} 2021.
\newblock \href {https://doi.org/10.1145/3653014} {\path{doi:10.1145/3653014}}.

\bibitem{GaoIKW19}
Jiawei Gao, Russell Impagliazzo, Antonina Kolokolova, and Ryan Williams.
\newblock Completeness for first-order properties on sparse structures with
  algorithmic applications.
\newblock {\em {ACM} Trans. Algorithms}, 15(2):23:1--23:35, 2019.
\newblock Announced at {SODA} 2017.
\newblock \href {https://doi.org/10.1145/3196275} {\path{doi:10.1145/3196275}}.

\bibitem{GoldreichR18}
Oded Goldreich and Guy~N. Rothblum.
\newblock Counting t-cliques: Worst-case to average-case reductions and direct
  interactive proof systems.
\newblock In {\em 59th {IEEE} Annual Symposium on Foundations of Computer
  Science, {FOCS} 2018}, pages 77--88. {IEEE} Computer Society, 2018.
\newblock \href {https://doi.org/10.1109/FOCS.2018.00017}
  {\path{doi:10.1109/FOCS.2018.00017}}.

\bibitem{HS23}
Shuichi Hirahara and Nobutaka Shimizu.
\newblock Hardness self-amplification: Simplified, optimized, and unified.
\newblock In {\em Proceedings of the 55th Annual {ACM} Symposium on Theory of
  Computing, {STOC} 2023}, pages 70--83. {ACM}, 2023.
\newblock \href {https://doi.org/10.1145/3564246.3585189}
  {\path{doi:10.1145/3564246.3585189}}.

\bibitem{ImpagliazzoPZ01}
Russell Impagliazzo, Ramamohan Paturi, and Francis Zane.
\newblock Which problems have strongly exponential complexity?
\newblock {\em J. Comput. Syst. Sci.}, 63(4):512--530, 2001.
\newblock Announced at {FOCS 1998}.
\newblock \href {https://doi.org/10.1006/JCSS.2001.1774}
  {\path{doi:10.1006/JCSS.2001.1774}}.

\bibitem{ImpagliazzoR89}
Russell Impagliazzo and Steven Rudich.
\newblock Limits on the provable consequences of one-way permutations.
\newblock In {\em Proceedings of the 21st Annual {ACM} Symposium on Theory of
  Computing}, pages 44--61. {ACM}, 1989.
\newblock \href {https://doi.org/10.1145/73007.73012}
  {\path{doi:10.1145/73007.73012}}.

\bibitem{JuelsP00}
Ari Juels and Marcus Peinado.
\newblock Hiding cliques for cryptographic security.
\newblock {\em Designs, Codes and Cryptography}, 20(3):269--280, 2000.
\newblock Announced at {SODA} 1998.

\bibitem{KaneW19}
Daniel~M. Kane and Richard~Ryan Williams.
\newblock The orthogonal vectors conjecture for branching programs and
  formulas.
\newblock In {\em 10th Innovations in Theoretical Computer Science Conference,
  {ITCS} 2019}, volume 124 of {\em LIPIcs}, pages 48:1--48:15. Schloss Dagstuhl
  -- Leibniz-Zentrum f{\"{u}}r Informatik, 2019.
\newblock \href {https://doi.org/10.4230/LIPICS.ITCS.2019.48}
  {\path{doi:10.4230/LIPICS.ITCS.2019.48}}.

\bibitem{KarppaKK18}
Matti Karppa, Petteri Kaski, and Jukka Kohonen.
\newblock A faster subquadratic algorithm for finding outlier correlations.
\newblock {\em {ACM} Trans. Algorithms}, 14(3):31:1--31:26, 2018.
\newblock Announced at {SODA} 2016.
\newblock \href {https://doi.org/10.1145/3174804} {\path{doi:10.1145/3174804}}.

\bibitem{KarppaKKC20}
Matti Karppa, Petteri Kaski, Jukka Kohonen, and Padraig~{\'{O}} Cath{\'{a}}in.
\newblock Explicit correlation amplifiers for finding outlier correlations in
  deterministic subquadratic time.
\newblock {\em Algorithmica}, 82(11):3306--3337, 2020.
\newblock Announced at {ESA} 2016.
\newblock \href {https://doi.org/10.1007/S00453-020-00727-1}
  {\path{doi:10.1007/S00453-020-00727-1}}.

\bibitem{SM19}
{Karthik {C. S.}} and Pasin Manurangsi.
\newblock On closest pair in euclidean metric: Monochromatic is as hard as
  bichromatic.
\newblock In {\em 10th Innovations in Theoretical Computer Science Conference,
  {ITCS} 2019}, volume 124 of {\em LIPIcs}, pages 17:1--17:16. Schloss Dagstuhl
  - Leibniz-Zentrum f{\"{u}}r Informatik, 2019.
\newblock \href {https://doi.org/10.4230/LIPICS.ITCS.2019.17}
  {\path{doi:10.4230/LIPICS.ITCS.2019.17}}.

\bibitem{LaVigneLW19}
Rio LaVigne, Andrea Lincoln, and Virginia {Vassilevska Williams}.
\newblock Public-key cryptography in the fine-grained setting.
\newblock In {\em Advances in Cryptology -- {CRYPTO} 2019 -- 39th Annual
  International Cryptology Conference}, volume 11694 of {\em Lecture Notes in
  Computer Science}, pages 605--635. Springer, 2019.
\newblock \href {https://doi.org/10.1007/978-3-030-26954-8\_20}
  {\path{doi:10.1007/978-3-030-26954-8\_20}}.

\bibitem{Merkle78}
Ralph~C. Merkle.
\newblock Secure communications over insecure channels.
\newblock {\em Commun. {ACM}}, 21(4):294--299, 1978.
\newblock \href {https://doi.org/10.1145/359460.359473}
  {\path{doi:10.1145/359460.359473}}.

\bibitem{Rosen20}
Alon Rosen.
\newblock Fine-grained cryptography: {A} new frontier?
\newblock {\em {IACR} Cryptol. ePrint Arch.}, page 442, 2020.
\newblock URL: \url{https://eprint.iacr.org/2020/442}.

\bibitem{Valiant15}
Gregory Valiant.
\newblock Finding correlations in subquadratic time, with applications to
  learning parities and the closest pair problem.
\newblock {\em J. {ACM}}, 62(2):13:1--13:45, 2015.
\newblock Announced at {FOCS} 2012.
\newblock \href {https://doi.org/10.1145/2728167} {\path{doi:10.1145/2728167}}.

\bibitem{VVW18}
Virginia Vassilevska~Williams.
\newblock {\em On some fine-grained questions in algorithms and complexity},
  pages 3447--3487.
\newblock World Scientific, 2018.
\newblock \href {https://doi.org/10.1142/9789813272880_0188}
  {\path{doi:10.1142/9789813272880_0188}}.

\bibitem{Williams19}
R.~Ryan Williams.
\newblock Some estimated likelihoods for computational complexity.
\newblock In {\em Computing and Software Science -- State of the Art and
  Perspectives}, volume 10000 of {\em Lecture Notes in Computer Science}, pages
  9--26. Springer, 2019.
\newblock \href {https://doi.org/10.1007/978-3-319-91908-9\_2}
  {\path{doi:10.1007/978-3-319-91908-9\_2}}.

\bibitem{Williams05}
Ryan Williams.
\newblock A new algorithm for optimal 2-constraint satisfaction and its
  implications.
\newblock {\em Theor. Comput. Sci.}, 348(2-3):357--365, 2005.
\newblock \href {https://doi.org/10.1016/J.TCS.2005.09.023}
  {\path{doi:10.1016/J.TCS.2005.09.023}}.

\end{thebibliography}

\end{document}